\newcommand{\citet}{\cite}
\newcommand{\citealt}{\cite}

\newcommand{\citep}{\cite}

\newcommand{\titleText} {\bf Fundamental Properties of the Evolution of Mutational Robustness}
\newcommand{\titleHeader} {Fundamental Properties of the Evolution of Mutational Robustness}
\newcommand{\abstractText}{Evolution on neutral networks of genotypes has been found in models to concentrate on genotypes with high mutational robustness, to a degree determined by the topology of the network.  Here analysis is generalized beyond neutral networks to arbitrary selection and parent-offspring transmission.  In this larger realm, geometric features determine mutational robustness:  the alignment of fitness with the orthogonalized eigenvectors of the mutation matrix weighted by their eigenvalues.  ``House of cards'' mutation is found to preclude the evolution of mutational robustness.  Genetic load is shown to increase with increasing mutation in arbitrary single and multiple locus fitness landscapes.  The rate of decrease in population fitness can never grow as mutation rates get higher, showing that ``error catastrophes'' for genotype frequencies never cause precipitous losses of population fitness.  The ``inclusive inheritance'' approach taken here naturally extends these results to a new concept of dispersal robustness.
}
\newcommand{\authorText} {Lee Altenberg \affil{1}{The Konrad Lorenz Institute for Evolution and Cognition Research, Martinstrasse 12, Klosterneuburg, A3400 Austria, {Lee.Altenberg@kli.ac.at}
}
}

\newcommand{\keyText}{genetic load | spectral gap | lethal mutagenesis | epigenetic mutation | dispersal load}
\newcommand{\acknowledgeText}{Insight for this paper occurred while hearing Ludwig Geroldinger's dissertation defense on stepping stone and island migration models \citep{Geroldinger:and:Burger:2015:Clines}.
I thank Erik van Nimwegen, Reinhard B\"urger, Joachim Hermisson, and Nick Barton for their insightful comments.  I gratefully acknowledge support from The KLI Institute, Austria, and the Mathematical Biosciences Institute at Ohio State University, USA, through National Science Foundation Award \#DMS 0931642.
}

\documentclass[twocolumn,10pt]{article}	

\usepackage{amscd}
\usepackage{amsmath}
\usepackage{amsthm}
\usepackage{amsfonts}
\usepackage{amssymb}
\usepackage{cases}		 
\usepackage{enumerate}
\usepackage{graphicx}
\usepackage{hyperref}
\usepackage{mathrsfs}		 
\usepackage{mathtools}
\usepackage{url}
\usepackage{wrapfig}

\newcommand{\hide}[1]{}	
\hide{
\textwidth=6.8in	
\oddsidemargin=-0.2in	
\evensidemargin=-0.2in	
\columnsep=0.15in	
\textheight=8in
}

\newcommand{\D}{\bm{D}}

\newcommand{\I}{\bm{I}}
\newcommand{\K}{\bm{K}}
\newcommand{\Lam}{\bms{\Lambda}}

\newcommand{\M}{\bm{M}}
\newcommand{\NCH}{{Nimwegen:Crutchfield:and:Huynen:1999}}
\newcommand{\Nc}{{\cal N}}
\newcommand{\Ox}{\bigotimes}

\newcommand{\ab}[1]{\begin{align*}#1\end{align*}}

\newcommand{\an}[1]{\begin{align}#1\end{align}}
\newcommand{\av}{\bm{a}}
\newcommand{\bms}[1]{{\boldsymbol{#1}}}
\newcommand{\bm}[1]{{\bf #1}}
\newcommand{\cv}{\bm{c}}
\newcommand{\df}[2]{\displaystyle \frac{\mbox{\rm d} #1}{\mbox{\rm d} #2}}
\newcommand{\diag}[1]{\mbox{ \bf diag}\matrx{#1}}
\newcommand{\dspfrac}[2]{\frac{\displaystyle #1}{\displaystyle #2} }

\newcommand{\eqdef}{:=}
\newcommand{\ev}{\bm{e}}
\newcommand{\evt}{{\,\ev\tr}}

\newcommand{\matrx}[1]{{\left[ \stackrel{}{#1}\right]}}
\newcommand{\ov}{\overline}

\newcommand{\piv}{\bms{\pi}}

\newcommand{\tr}{^{\!\top}}

\newcommand{\x}{\bm{x}}

\newcommand{\y}{\bm{y}}
\newcommand{\desclist}[1]{\begin{description}#1\end{description}}
\newcommand{\enumlist}[1]{\begin{enumerate}#1\end{enumerate}}

\newcommand{\sumin}{\sum_{i=1}^n}	
\newcommand{\sumjn}{\sum_{j=1}^n}
\newcommand{\sumijn}{\sum_{i,j=1}^n}
\newcommand {\dsty}{\displaystyle}
\newcommand{\stext}{\shortintertext}	
\newcommand{\0}{{\bm 0}}
\newcommand{\pf}[2]{\displaystyle \frac{\partial #1}{\partial #2}}

\newcommand{\suchthat}{\colon}

\renewcommand{\c}{\bm{c}}
\newcommand{\OxL}{\Ox_{\xi=1}^L}
\newcommand{\ddf}[2]{\displaystyle \frac{\mbox{\rm d}^2 #1}{\mbox{\rm d} #2^2}}
\newcommand{\Reals}{\mathbb{R}} 

\newcommand{\zh}{\hat{z}}
\newcommand{\zvh}{{\hat{\z}}}

\newcommand{\z}{\bm{z}}

\newcommand{\Bmatr}[1]{\begin{bmatrix*}[r]\dsty #1\end{bmatrix*}}

\newcommand{\wb}{\ov{w}}
\newcommand{\wbh}{\widehat{\wb} }

\newcommand{\wv}{\bm{w}}
\newcommand{\Cov}{\mbox{\rm Cov}}

\newcommand{\Lcu}{\mathscr{L}}

\newcommand{\goesto}{\rightarrow}

\newcommand{\w}{\bm{w}}
\newcommand{\Oc}{{\cal O}}
\newcommand{\Pc}{{\cal P}}

\newcommand{\nubh}{\hat{\ov{\nu}}}
\newcommand{\fracinline}[2]{{#1 / #2}}
\newcommand{\Var}{{\sf Var}}

\newcommand{\sxi}{^{(\xi)}}  

\newcommand{\AMR}{{\sf Ro}}	
\newcommand{\RMR}{{\sf RRo}}	
\newcommand{\SI}{\emph{SI}}
\newcommand{\Ex}[1]{{\mathbb{E}[#1]}} 
\newcommand{\Sm}{\bm{S}}

\renewcommand{\P}{\bm{P}}
\newcommand{\mut}{\text{\sf mu}}

\usepackage[usenames]{color} 
\newcommand{\affil}[2]{\footnote{#2}}

\newtheorem{Theorem}{Theorem}
\newtheorem{Definition} [Theorem]{Definition}
\newtheorem{Result} [Theorem]{Result}
\newtheorem{Corollary}[Theorem]{Corollary}
\newtheorem{Lemma}[Theorem]{Lemma}
\newtheorem{RemarkNum}[Theorem]{Remark}

\author{\authorText}
\title{\titleText}

\begin{document}
\date{}
\maketitle
\abstract{
\bf \abstractText
}

\sloppy
\raggedbottom
Based on theoretical considerations, Kimura \citep{Kimura:1968:Evolutionary} predicted that the majority of evolutionary changes in the genome in mammals should consist of neutral mutations.  Since this time, research has been focused on understanding the extent and nature of neutral genetic variation in organisms.  One approach is to attempt to derive them from molecular first principles \citep{King:and:Jukes:1969,Hietpas:et:al:2011:Experimental}.    The idea that neutrality may not merely be a derived consequence of molecular interactions, but actually an evolved property shaped by evolutionary dynamics, had its first manifestation in Fisher's theory for the evolution of dominance \citep{Fisher:1928:Possible}, followed by  Waddington \citep{Waddington:1942} who proposed it as a consequence of stabilizing selection, and Conrad (cf. \citep{Conrad:1974:Molecular-I,Conrad:1979:Mutation}) who proposed it could  result from higher-order epistatic mutations which smooth the adaptive landscape at other sites and consequently enhance evolvability.  Another mechanism proposed is that natural selection for adaptations robust to {environmental} variation entails robustness to {mutation} as a generic side-effect \citep{Ancel:and:Fontana:2000}.

An altogether different mechanism for the evolution of mutational robustness was proposed by Bornberg-Bauer and Chan \citet{Bornberg:and:Chan:1999:Modeling} and van Nimwegen et al. \cite{Nimwegen:1999, Nimwegen:Crutchfield:and:Huynen:1999}, which is that evolution along \emph{neutral networks}---sets of mutationally connected genotypes with equivalent fitnesses---would ``concentrate at highly connected parts of the network, resulting in phenotypes that are relatively robust against mutations''\citep{Nimwegen:Crutchfield:and:Huynen:1999}, and ``Independent of functional fitness, topology \emph{per se} can lead to concentration of evolutionary population at some sequences.''\citep{Bornberg:and:Chan:1999:Modeling}.

Numerous citations of these papers repeat the finding that ``populations will evolve toward highly connected regions of the genome space'' (e.g. \citet{Aguirre:Buldu:and:Manrubia:2009}, \citet{Kun:etal:2015:Dynamics}).  However, neither the original papers nor subsequent studies (to my knowledge) provide analytic proofs of this observation.  With few exceptions \citep{Reeves:etal:2015:Eigenvalues} progress has been limited  in identifying exactly which properties of neutral networks determine their evolved mutational robustness.  As is shown by the following example, more is going on than simply a ``tendency to evolve toward highly connected parts of the network''.

Figure \ref{fig:TwoTopologies} compares the equilibrium population distribution for two neutral networks that have 63 equally fit genotypes, and one inviable genotype off the network.  The only difference between the networks is their mutational topology.  One network is the set of all 64 nucleotide triplets, while the other is the set of copy-number variants, from $1$ to $64$ copies.  
Two properties of the population equilibrium distributions on the neutral networks stand out:  first, the equilibrium frequency of a genotype is determined not solely by how many neutral neighbors it has, but also by its position within the network.  In the copy-number network, 62 of the genotypes are identical in having only neutral neighbors.  Yet the stationary distribution increases 20-fold over these 62 genotypes as they get more mutationally distant from the lethal genotype \citep{Altenberg:2005:Evolvability}.  The same phenomenon is seen in the trinucleotide network but to a lesser degree.  Secondly, dramatic differences are seen between the two neutral networks in the size of the genetic loads they maintain.  The genetic load of the trinucleotide network is 44 times the genetic load of the copy-number network.  

\begin{figure}[t]
\begin{center}
\includegraphics[width=\columnwidth]{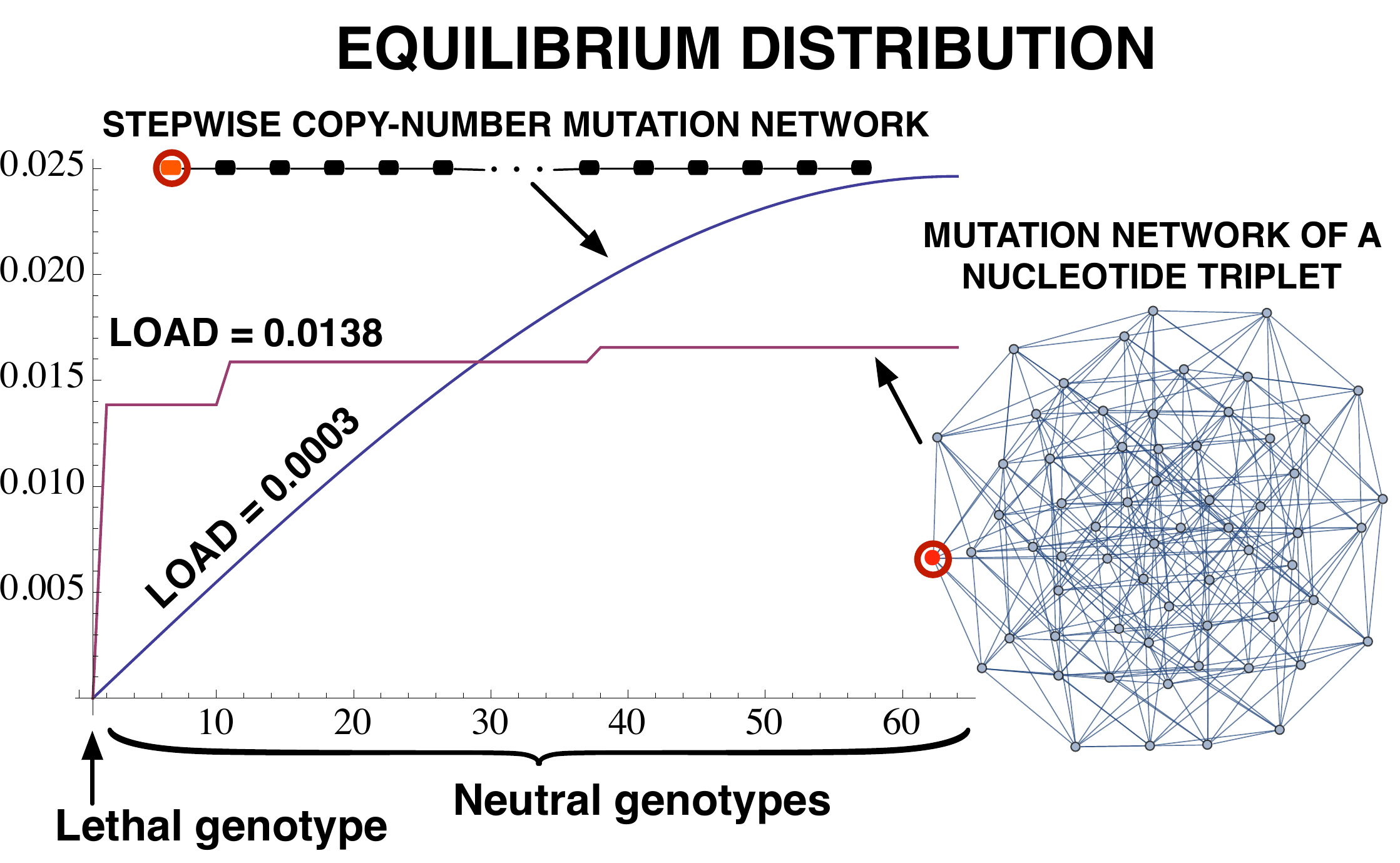}
\end{center}
\caption{ \label{fig:TwoTopologies} Equilibrium distributions and genetic loads for two neutral networks having 63 neutral genotypes and one lethal genotype (circled), under two mutation topologies: Mutation of a nucleotide triplet, and stepwise copy-number mutation.  
}
\end{figure}

The differences between the evolutionary outcomes on these two mutational graphs can only be the result of their different topologies, but what properties of their topologies?  Numerous results from the field of spectral graph theory are applicable to this question; here no review is attempted.  While graph theory has been widely used in models of mutation, the results typically impose the assumptions that mutations occur only once during replication, mutation is symmetric and occurs at a single rate, and fitness is either zero or a single other value (all assumptions in \citet{Bornberg:and:Chan:1999:Modeling,Nimwegen:Crutchfield:and:Huynen:1999}).

An alternative approach is taken here, which has proven valuable in previous work \citep{Altenberg:1984:Book,Altenberg:and:Feldman:1987,Altenberg:2010:Feldman-Karlin,Altenberg:2012:Resolvent-Positive,Altenberg:2012:Dispersal}.  The specific problem---evolution on neutral networks in nucleotide sequence space---is embedded into the larger space of problems, which  includes arbitrary mutation patterns and arbitrary selection values.   The generality forces one to seek the appropriate fundamental mathematical properties.  

It also expands the applicability of any results.  Mutation is treated generally enough to apply to non-genetic information transmission, the principal example being organismal location, where the analog of mutation is dispersal.  The results here thus automatically apply to {dispersal load}, and in the process define a new concept of \emph{dispersal robustness}.  Moreover, the generality of the treatment has the potential to apply widely to diverse mechanisms of ``inclusive inheritance'' \citep{Danchin:and:WagnerR:2010:Inclusive}.

With arbitrary fitnesses, one can no longer characterize mutation neighborhoods simply by the fraction of mutations that are neutral, since now the distribution of fitness effects of mutation (DFE) includes advantageous or deleterious mutations, as well as neutral and lethal.
The more general statistic is the expected fitness of offspring.  Averaged over the population, one obtains an aggregate \emph{population mutational robustness} --- the degree to which the population maintains its fitness in the face of mutation pressure.  A complete absence of mutational robustness occurs when the genetic load is the mutation rate, corresponding to the Haldane-Muller principle \citep{Haldane:1937:Effect,Muller:1950}.   Complete mutational robustness, on the other hand, would mean that the population suffers no genetic load as the mutation rate increases.   A given adaptive landscape will fall somewhere between these two extrema.  

The main results found here are that the population mutational robustness at a mutation-selection balance is determined by abstract spectral properties:  the alignment between the fitnesses and the eigenvectors of the mutation matrix, weighted by their eigenvalues.  

This spectral analysis provide a lens through which to examine models of mutation, dispersal, an inheritance generally.  Models of mutation and selection are usually constructed without appreciation of how the assumptions manifest in the eigenvalues and eigenvectors of the mutation matrix (or variation operator in the case of continuous variation).  We will see, for example, that the widely encountered \citep{Hermisson:Hansen:and:Wagner:2003:Epistasis} ``house-of-cards'' mutation model \citet{Kingman:1977:Properties} is incapable of supporting mutational robustness.  The results here provide direction for analyzing a wide variety of models for their capacity to support the evolution of mutational robustness, and enable comparisons to be made between different kinds of mutation processes, including nucleotide base mutation, epigenetic mark mutation, and gene copy number variation.  The comparisons extend to the spectral properties of dispersal matrices, thereby unifying the results for genetic robustness and genetic load with those for ``dispersal robustness'' and dispersal load.  

\section{The Setting}

The object of interest is the state of the population when it has converged in frequencies to an equilibrium under the following asexual, haploid evolutionary dynamics, where the population is large enough to be treated as infinite and has discrete non-overlapping generations (semelparity).  The only event that changes the genotypes during reproduction is mutation; there is no recombination.
\an{
z_i(t+1) &= \frac{1}{\wb(t)}\ \sumjn M_{ij}(\mu) \, w_j z_j(t), \qquad \label{eq:Haploid}
}
\text{where} \desclist {
\item[$n$] is the number of possible haploid genotypes,
\item[$z_i(t)$] is the frequency of haploid genotype $i$ in the population at time $t$, $z_i(t) \geq 0$, $\sumin z_i(t) = 1$,
\item[$w_j$] is the fitness of genotype $j$, $w_i \geq 0$,
\item[$M_{ij}(\mu)$] is the probability, when the mutation rate is $\mu \in [0,1]$, that parent of genotype $j$ has offspring of genotype $i$, $M_{ij} \geq 0$,  $\sumin M_{ij} = 1$ $\forall j$, and 
\item[$\wb(t)$] $= \sumin w_i z_i(t)$ is the mean fitness of the population at time $t$, its Malthusian rate of increase in size.
}
In the case of a multilocus genotype where mutation occurs independently at each of $L$ loci, the mutation probabilities may be decomposed as products of the mutation probabilities at each locus,
$M_{ij} = M_{i_1 j_1}^{(1)} M_{i_2 j_2}^{(2)} \cdots M_{i_L j_L}^{(L)}$,
where $i_\xi, j_\xi$ index the alleles at locus $\xi$.

Mutation at each locus $\xi$ may be parameterized by a mutation rate $\mu$ and a \emph{mutation distribution}, $P_{i_\xi j_\xi}^{(\xi)}$, given that mutation occurs: $M_{i_\xi j_\xi}^{(\xi)} = (1-\mu) \delta_{i_\xi j_\xi} + \mu \, P_{i_\xi j_\xi}^{(\xi)}$,
where $\delta_{ii} = 1$ and $\delta_{ij} = 0$ when $i \neq j$.  In matrix form this is
\an{
\M\sxi = (1-\mu) \I\sxi + \mu \P\sxi, \label{eq:MxiP}
}
where each $\P\sxi$ is an  $n_\xi \times n_\xi$ column-stochastic matrix, $n_\xi$ is the number of possible alleles at locus $\xi$, and $\I\sxi$ is the $n_\xi \times n_\xi$ identity matrix.

The multilocus mutation matrix can be represented using the Kronecker product, $\Ox$, as
\an{
\M(\mu) = \Ox_{\xi=1}^L \left[(1-\mu) \I^{(\xi)} + \mu \P^{(\xi)}\right].\label{eq:MultiMuPMat}
}

Matrices of the form $(1-\mu) \I + \mu \P$ represent the situation where only a \emph{single transforming event} occurs during reproduction and will be referred to as \emph{single-event mutation} matrices, while those of the form \eqref{eq:MultiMuPMat} represent the situation where \emph{multiple independent transforming events} occur during reproduction and will be referred to as \emph{multiple-event mutation} matrices.  In the limit of small $\mu$, multiple-event matrices converge to single-event matrices, even when they model multiple loci. 

From \eqref{eq:Haploid}, equilibrium frequencies must satisfy $\wbh \zh_i = \sumjn M_{ij} w_j \zh_j $, or in vector form,
\an{ 
\M(\mu) \, \D \, \zvh &= (\w\tr \zvh) \, \zvh 
=\wbh \ \zvh, \qquad  \label{eq:EqIdVec}
}
\text{where}\desclist {
\item[$\ev$] is the vector of ones, $\ev\tr = \Bmatr{1&1&\cdots&1}$ its transpose, 
\item[$\D$]$ \eqdef \diag{w_i}$ is the diagonal matrix of fitness coefficients, 
\item[$\w$]$=\D\ev$ is the vector of fitness coefficients, and
\item[$\wbh$] $\eqdef\sumin w_i \zh_i =  \ev\tr \D \zvh = \w\tr\zvh$ is the mean fitness of the population. 
}
Equation \eqref{eq:EqIdVec} shows that the mean fitness $\wbh = \ev\tr \D \zvh$ is an eigenvalue of $\M(\mu) \D$.

It is further assumed that $\M(\mu)$ is irreducible when $\mu > 0$, meaning any genotype $j$ can mutate in some number of steps to any other genotype $i$.  $\M(\mu)\D$ is irreducible as well if no genotype is lethal, so from Perron-Frobenius theory we know there is only one possible equilibrium, the strictly positive \emph{Perron vector} $\zvh$, and eigenvalue $\wbh$ 
is the \emph{Perron root} and spectral radius of $\M(\mu) \D$, referred to as $r(\M(\mu)\D) = \wbh$.

If some mutations are lethal, then $\M(\mu) \D$ is reducible, and more interesting equilibrium behavior becomes possible.  The Frobenius normal form of a reducible matrix (\citet[pp. 74-77]{Gantmacher:1959:2}) partitions the set of genotypes into blocks, such that the restriction of $\M(\mu)\D$ to any block $\kappa$ produces an irreducible matrix $\M(\mu)^{[\kappa]} \D^{[\kappa]}$.  The blocks constitute \emph{quasispecies} \citep{Eigen:and:Schuster:1977}.  The block with the largest spectral radius $r(\M(\mu)^{[\kappa]} \D^{[\kappa]})$ among the genotypes present at equilibrium sets the mean fitness of the population.  The analysis will therefore focus on such irreducible block matrices.

\subsection{Diploidy, frequency-dependent selection, recombination, assortative mating, and dispersal}
The mutational robustness of population at equilibrium can be analyzed without knowing whether this equilibrium is dynamically stable.  In the case of haploid selection on a quasispecies without recombination, the system is essentially linear and the polymorphic equilibrium is always stable.  Relaxation of these assumptions to accommodate greater biological variety may alter the stability and existence conditions of equilibria but does not change the robustness analysis.  The results here therefore apply to equilibria in systems with diploid and frequency-dependent selection, assortative mating, and also recombination.  In this case, letting $T_{ijk}$ be the probability that parents $j$ and $k$ produce offspring $i$, one gets a a frequency-dependent version of \eqref{eq:Haploid}:
\an{
z_i(t{+}1) & = \frac{1}{\wb(t)} \sumjn z_j(t) \sum_{k=1}^n  T_{ijk}  w_{jk} X_{jk}(\z), \label{eq:Diploid} \\
&= \frac{1}{\wb(t)} \sumjn M_{ij}(\z) \, w_j(\z) \,  z_j(t), \label{eq:DiploidFDS}
}
where $w_{jk}$ is the fitness of diploid genotype $jk$, $X_{jk}(\z(t))$ is the probability that the mate of genotype $j$ is genotype $k$ (frequency dependent), $w_j(\z)  
= \sum_{k=1}^n  w_{jk} X_{jk}(\z)$, and 
$
M_{ij}(\z)  = \sum_{k=1}^n  T_{ijk}  w_{jk} X_{jk}(\z)/w_j(\z).
$
Under random mating, $X_{jk}(\z(t)) = z_k(t)$.  Selfing cannot be accommodated within the product structure in \eqref{eq:Diploid}.   
In models of dispersal, $\z(t)$ typically represents the density of organisms, without normalization  by $\wb$. Carrying capacities then makes the growth rates $w_i(\z)$ density dependent. 

Because frequency dependence is relevant only to stability, which is not considered, the analysis to follow applies to both \eqref{eq:Haploid} and \eqref{eq:DiploidFDS}.  However, frequency dependence of both $M_{ij}(\z)$ and $w_j(\z)$ makes them composite formal quantities, rather than biological essential quantities as in the case of mutation and haploid selection. 

\subsection{Mean Fitness, Genetic Load, and Mutational Robustness}

Let us retrace the calculation of population mutational robustness in the neutral network model of van Nimwegen et al.\ \citep{\NCH}.  Genotypes have only two possible fitnesses, $w$, and $0$.  The mutational robustness of a genotype with fitness $w$ is the probability that its offspring have fitness $w$.  The set of genotypes in the neutral quasispecies is $\Nc$.  Then the mutational robustness of a genotype---it's \emph{neutrality}--- is $\nu_j \eqdef \sum_{i \in \Nc} M_{ij}$.
The population neutrality at equilibrium $\zvh$ is the expectation of $\nu_j$ censused before reproduction, with parent frequencies $w \zh_i / \wbh$ for $i \in \Nc$.   Using \eqref{eq:EqIdVec}, since $w_j=0$ or $\zh_j=0$ for $j \notin \Nc$,
\ab{
\zh_i &= \! \sumjn \! M_{ij} \frac{w_j  \zh_j}{\wbh } 
\!= \! \sum_{j \in \Nc} \! M_{ij} \frac{w  \zh_j}{\wbh } , \  
\wbh \!= \! \sum_{i=1}^n w_i \zh_i = \! \sum_{i \in \Nc} w \zh_i .
}
From these we see that the equilibrium average neutrality is 
\an{
\nubh &= \sum_{j \in \Nc} \nu_j \frac{w \zh_j }{ \wbh}
=  \sum_{i, j \in \Nc}  M_{ij}\frac{w \zh_j }{ \wbh} 
{=}  \sum_{i \in \Nc}  \zh_i 
{=} \frac{\wbh}{w} =\frac{r(\M\D)}{w}. \label{eq:AvgNeutrality}
}

The concept of mutational robustness can be extended beyond neutral networks by simply generalizing $w$ to $\max_i[w_i]$.
\begin{Definition}
At a mutation-selection balance, the {\bf population mutational robustness} is defined to be \an{
\AMR(\M, \D) \eqdef \frac{\wbh} {\max_i[w_i]} = \frac{r(\M\D)}{\max_i[w_i]} .\label{eq:AMRdef}
}
\end{Definition}
The quantity $ \fracinline{r(\M\D)}{\max_i[w_i]}=\AMR(\M, \D)=1 - {\sf L} $ is the complement of the classical {genetic load}, ${\sf L}$  \citep {Haldane:1937:Effect,Muller:1950}.  We would like to know where $\AMR(\M, \D)$ falls within the possible range of values given the mutation rate in $\M$.  The maximum possible value is $\AMR(\M, \D) = 1$, meaning no loss of fitness from mutation.  The minimum possible value will be called the \emph{Haldane limit}, which is $(1-\mu)^L$ for $L$ loci.

\begin{Result}[Haldane Limits on Mutational Robustness]\label{Result:Haldane}
\ab{
\AMR(\M(\mu),\D) &\eqdef \dspfrac{r(\M(\mu)\D)}{\max_i[w_i]} \geq
(1 - \mu)^L . 
}

\end{Result}
See also \citet[eq. (2.3)]{Burger:and:Hofbauer:1994:Mutation}, 
\citet[pp. 149--150, eq. (5.31)--(5.34)]{Burger:2000}.
The proof for this result and those following are provided in the {\bf Supporting Information} (\emph{SI}).

Robustness \emph{relative} to the Haldane limit can be defined.
\begin{Definition}
Define \emph{\bf relative mutational robustness} as
\ab{
&\RMR(\M(\mu), \D)  \eqdef \frac{\dspfrac{r(\M(\mu)\D)}{\max_i[w_i]} - (1-\mu)^L}{1 - (1-\mu)^L} \in [0, 1].
}
\end{Definition}

One further quantity of interest is the average fitness of \emph{mutant} offspring, 
$ 
\wb^{\mut}_j \eqdef { \sum_{i \neq j}^n w_i M_{ij}}/{ \sum_{i \neq j}^n M_{ij}}. 
$ 
\begin{Theorem}\label{Theorem:Mutant}
The average mutant offspring fitness in a population at mutation-selection balance is
\ab{
\wb^{\mut} (\M(\mu), \D) &\eqdef \sumjn \wb^{\mut}_j \frac{w_j \zh_j}{\wbh } \\&
= \wbh \left[1 - \frac{(1-\mu)^L}{1-(1-\mu)^L} \Var\big[\frac{w_i}{\wbh} \big]  \right],
}
assuming the mutation distributions have $P^{(\xi)}_{ii} {=} 0$ for all loci $\xi$ and alleles $i$. 
\end{Theorem}
This identity  precisely expresses Fisher's ``deterioration of the environment'' argument \citep[p. 42]{Fisher:1930}, that at any equilibrium, the average mutant fitness loss must exactly offset the gain in fitness from the fitness variance.

\section{Basic Results}
The core finding of \citet{\NCH} is that when genotype neutralities vary over the neutral network, ``the population neutrality is typically larger than the network neutrality. \ldots Thus, a population will evolve a mutational robustness that is larger than if the population were to spread uniformly over the neutral network.''
While numerical examples are explored, no analytical results are provided that define precisely when it holds.  It is claimed in \citet[eq. (12)]{Aguirre:Buldu:and:Manrubia:2009} that when the mutation matrix is symmetric, i.e.\ $\M = \M\tr$, this result can be obtained from Perron-Frobenius theory, but no proof is provided.  A proof is here provided (in \SI), but it requires more than Perron-Frobenius theory, specifically Rayleigh theory.

\begin{Theorem}[Neutral Network Robustness]\label{Result:NNR}
The equilibrium population robustness $\nubh = r(\M\D) / w$ is always greater than the average robustness over the neutral network, 
\ab{
\Ex{\nu} 
= \frac{\sum_{i \in \Nc} \sum_{j \in \Nc} M_{ij}}{\sumin \sum_{j \in \Nc} M_{ij}}.
}
if (1) the mutation matrix $\M$ is symmetric, (2) fitness off the network is $0$, and (3) not all genotypes on the neutral network have the same mutational robustness.
\end{Theorem}

When mutation is not symmetric, counterexamples can easily be constructed where the population evolves to an average robustness that is less than the average robustness over the entire neutral network.  One example is where mutation is cyclic.  Another example is where asymmetric mutation is symmetrizable (a reversible Markov chain), but mutation is biased toward genotypes with few neutral neighbors (in \SI).

The next result adopts to quasispecies a theorem of Karlin, which has been available for over 30 years, and was independently proven in \citet{Kirkland:Li:and:Schreiber:2006}.
\begin{Theorem}[5.2 of Karlin \citet{Karlin:1982}]
\label{Theorem:SingleLoad}
The equilibrium genetic load within any quasispecies under single-event mutation strictly increases with mutation rate when there is any variation in fitness, i.e.  for each irreducible block of loci $\kappa$, $r([(1-\mu)\I^{[\kappa]} + \mu \P^{[\kappa]}] \D^{[\kappa]})$ strictly decreases with $\mu$ when $\D^{[\kappa]} \neq c \ \I$ for any $c > 0$, or when $r(\P^{[\kappa]}) < 1$.
\end{Theorem}

The same outcome has been shown for multiple-event mutation when mutation at each locus is \emph{reversible}, by which I refer to $\M$ being the transition matrix of a reversible Markov chain, which is also equivalent to $\M$ being symmetrizable.
\begin{Theorem}[Corollaries 1, 3 of \citet {Altenberg:2011:Mutation}]\label{Theorem:MultiLoad}
The equilibrium genetic load of a quasispecies under irreducible multiple-event mutation strictly increases with mutation rate $0 < \mu < 1/2$, when the quasispecies has variation in fitness within its irreducible blocks.
\end{Theorem}

Theorems \ref{Theorem:SingleLoad} and \ref{Theorem:MultiLoad} have important implications for the theory of lethal mutagenesis.  They show that the population mean fitness will keep decreasing with greater mutation over at least the range $0 < \mu < 1/2$.   
\begin{Corollary}[Sufficiency for Lethal Mutagenesis]\label{Corollary:Lethal}
If $r(\M(1/2) \D) < 1$ then for some $\mu* < 1/2$, the population will go extinct for all $\mu > \mu*$.
\end{Corollary}

The observation in \citet{\NCH} that ``Perhaps surprisingly, the tendency to evolve toward highly
connected parts of the network is independent of evolutionary parameters---such as mutation rate,'' is solely due to there being only two fitnesses in the model, $0$, and $w$.  

\begin{Result}[Single-Event Relative Robustness Increases with Mutation Rates] \label{Result:Increases}
As the mutation rate increases with single-event mutation, if a quasispecies has more than one [only one] nonlethal fitness value, then the relative mutational robustness $\RMR(\M(\mu), \D)$ strictly increases [is constant] with mutation rate. 
\end{Result}
For multiple-event mutation, the relative mutational robustness may increase or decrease in $\mu$, as found in initial numerical exploration. 

\section{Results for Reversible Markov Chain Mutation}
Theorem \ref{Theorem:MultiLoad} was provable due to the tractability afforded by reversible Markov chains.  This same tractability carries over to the analysis of mutational robustness.  For a reversible chain, $\M$ can be represented (\citet[p. 33]{Keilson:1979}, \citet{Hermisson:etal:2002:Mutation}, \citet{Altenberg:2011:Mutation}) as: 
\ab{
\M = \D_{\piv}^{1/2}\K \Lam \K\tr \D_{\piv}^{-1/2}
}
where reversibility requires $M_{ij} \pi_j = M_{ji} \pi_j$, and
\desclist{
\item[$\piv$] $ = (\pi_1, \ldots, \pi_n)\tr = \M \piv$ is the stationary distribution for $\M$, i.e.\ its \emph{Perron vector},
\item[$\D_{\piv}^{1/2}$] is the diagonal matrix of the square roots of $\pi_i$,
\item[$\K$] is the matrix of orthogonalized eigenvectors of $\M$, $\K \K\tr = \K\tr\K = \I$,  i.e.\  for each $j=1, \ldots, n$, $\sum_{i=1}^n K_{ij}^2 = 1$, $\sum_{i=1}^n K_{ij} K_{ih} = 0$ if $j\neq h$, 
\item[\protect{$[\K]_j$}] is the $j$-th column of $\K$, 
\item[$K_{i1} = \pi_i^{1/2}$] or $[\K]_1 = \D_\piv^{1/2} \ev$,
\item[$\Lam$] $ \eqdef \diag{\lambda_i}$ is the diagonal matrix of eigenvalues  $\lambda_1 > \lambda_2 \geq \cdots \geq \lambda_n$ of $\M$.
}
Applying this representation of $\M$, we find that the spectral radius of $\M \D$ is that of a symmetric matrix:
\ab{
r(\M \D) &
= r(\D_{\piv}^{1/2}\K \Lam \K\tr \D_{\piv}^{-1/2} \D) \\&
= r(\K \Lam \K\tr \D_{\piv}^{-1/2} \D \D_{\piv}^{1/2}) \\&
= r(\K \Lam \K\tr \D) = r( \Lam \K\tr \D \K) \\&
= r(\D^{1/2}\K \Lam \K\tr  \D^{1/2}).
}
This allows use of the Rayleigh-Ritz variational formula for the spectral radius (\citealt[pp. 172--173]{Wilkinson:1965}, \citealt[pp. 176--180]{Horn:and:Johnson:1985}), to give the first main result.
\begin{Theorem}[Expression for the Spectral Radius]
When mutation has the transition matrix of a reversible Markov chain, $\M = \D_{\piv}^{1/2}\K \Lam \K\tr \D_{\piv}^{-1/2}$, then
\an{
r(\M\D) &= \max_{\x\tr\x = 1}  \x\tr \D^{1/2}\K \Lam \K\tr  \D^{1/2} \x \notag \\&
= \max_{\x\tr\x = 1} \sumjn \lambda_j \bigg(\sumin x_i \sqrt{w_i} K_{ij}\bigg)^2. \label{eq:rMD}
}
\end{Theorem}
Here we see a fundamental property of the population mutational robustness:
\begin{Corollary}
Holding $\K$ and $\w$, fixed, $r(\M \D)$ 
is non-decreasing in each eigenvalue of the mutation matrix, $\lambda_i$, $i=2, \ldots, n$.
\end{Corollary}

\subsection{Bounds Derived from the Mutational {Eigenvectors}}

Weinberger \citet{Weinberger:1991:Fourier} introduced the idea of representing fitness landscapes using the eigenvectors of the matrix representing the mutation structure, and this ``Fourier expansion'' approach has been elaborated by numerous further studies  (cf. \citet{Stadler:1996:Landscapes,
Klemm:and:Stadler:2014:Rugged}).  The second main result is now presented.
\begin{Theorem}{\bf (Lower Bound from the Fourier Representation of Fitness)}\label{Theorem:Main:Eigenvector}
Represent the fitnesses $w_i$ as $w_i =  \sumjn K_{ij} a_j$ or $\wv = \K \av$, where $a_i$ is the Fourier coefficient for eigenvector $[\K]_i$.  Then $$r(\M\D)\geq\frac{1}{\sumjn w_j} \sumjn \lambda_j a_j^2.$$
\end{Theorem}
The Fourier coefficients $a_j$ measure the `alignment' between $\w$ and each column $[\K]_j$, and $a_j^2$ is its ``amplitude'' \citep{Hordijk:and:Stadler:1998:Amplitude}.   From the theory of discrete nodal domains \citep{Davies:etal:2001:Discrete}, we know that as $j$ increases from $1$ to $n$, the eigenvectors $[K]_j$ exhibit more sign changes between their mutationally connected domains, which can be considered the `frequency' of the eigenvector, and contribute to ruggedness in the fitness landscape.  Greater ruggedness of the landscape (weight on the $a_j$ with larger $j$) corresponds to a smaller lower bound for the population mutational robustness, an effect that can appear even before there are multiple fitness peaks.

Because Theorem \ref {Theorem:Main:Eigenvector} applies to arbitrary fitnesses, it includes the mutational landscape modifiers described in \citet{Hermisson:etal:2002:Mutation}, early envisioned in \cite{Conrad:1979:Mutation}.  Note that the range of mutational robustness values possible from different arrangements of fitness values depends on the spread of the eigenvalues $\lambda_2, \ldots, \lambda_n$.

Prior applications of Fourier expansions of landscapes (e.g. \citet{Stadler:1996:Landscapes}, \citet{Reidys:and:Stadler:2001:Neutrality}) do not actually examine selection-mutation dynamics.  Rather, only mutation occurs, and fitnesses do not enter into reproduction but are simply recorded in each generation to produce a time series.  The main object of analysis has been the autocorrelation function of this time series when it is stationary, derived from the covariance between parent and offspring fitnesses.  The next result ties the mutation-only covariance to the mutation-selection dynamics.

\begin{Theorem}[Relation to Random Mutational Walks]
For a stationary mutational random walk with reversible Markov chain transition matrix $\M$, and stationary distribution $\piv$, let $\Cov_\piv(W_{\Pc}, W_{\Oc})$ be the $\piv$-weighted covariance between parent and offspring fitnesses.  Then at a mutation-selection balance, $\wbh \, \zvh = \M\D\zvh$,
\ab{
r(\M\D) \geq \left(\sumjn w_j \pi_j \right) \left[1 + \frac{1}{{(\w\tr\piv)}^2} \Cov_\piv[ {W_\Pc}, {W_\Oc}]\right].
}
\end{Theorem}
The expression states that the mean fitness of the population is increased above the $\pi$-weighted average fitness of the genotype space by the covariance between the normalized parent and offspring fitnesses in the mutational random walk.  

\begin{Corollary}
When the eigenvalues of $\M$ are positive, then $\Cov_\piv[ {W_\Pc}, {W_\Oc}] \geq 0$.
\end{Corollary}

\subsection{Bounds Derived from the Mutational Eigenvalues}
Now we consider bounds on $r(\M\D)$ that derive from the eigenvalues of $\M$ for any fitness landscape, to conclude the main results.
\begin{Theorem}[Upper Bound from the Spectral Gap] \label{Theorem:Main:Eigenvalue}
Assuming $0 < \mu < 1/2$, then
\ab{
\sumin w_i \pi_i 
&\leq r(\M\D) \\&
\leq \sumin w_i \pi_i + \lambda_2(\M) \left(\max_i[w_i]- \sumin w_i \pi_i\right),
}
where $\lambda_2(\M) = 1{-}\mu {+} \mu {\dsty \max_{\xi=1, \ldots, L}} \lambda_{2}(\P^{(\xi)})$.
\end{Theorem}
Theorem \ref{Theorem:Main:Eigenvalue} when applied to a neutral network at biological mutation rates (which give $\lambda_i > 0$ for all $i$) yields:
\begin{Corollary}[Spectral Gap Bound for Neutral Networks] \label{Result:Main:Neutral}
Let the neutral network be referred to as $\Nc = \{i : w_i = w > 0\}$.  Then
\an{
\sum_{i \in \Nc} \pi_i  \leq \frac{r(\M\D)}{w}
= \AMR(\M, \D) 
\leq \sum_{i \in \Nc} \pi_i + \lambda_2 \sum_{i \notin \Nc} \pi_i. \label{eq:GapNeutral}
}
\end{Corollary}
Here we see the explanation for the different mutational loads in Fig. \ref{fig:TwoTopologies}.  For the trinucleotide network, $\lambda_2= 5/9$, while for the copy-number variation network, $\lambda_2= 0.999$.  Eq. \eqref{eq:GapNeutral} gives lower bounds on load as $(1-\lambda_2)/64$.  Comparison of the bounds with the actual loads give $0.007 < 0.0138, 0.00002 < 0.0003$, respectively, for the two networks.

Next we shall see that the convexity of $r(\M(\mu)\D)$ in $\mu$ seen for single-event mutation in Result \ref {Result:Increases} extends to multiple-event reversible mutation.
\begin{Theorem}[Convexity of the Spectral Radius in $\mu$]\label {Theorem:ConvexMulti}
In the case where $L$ loci mutate independently at rate $\mu$, each in a reversible Markov chain, then for $0 < \mu <  1/2$, 
\ab{
&r(\M(\mu) \D)  = \\
\dsty &\max_{\x\tr\x = 1} 
\x\tr \D^{1/2} \bigg[\Ox_{\xi=1}^L \K^{(\xi)} [(1{-}\mu) \I^{(\xi)} {+} \mu \Lam^{(\xi)}]  {\K^{(\xi)}}\tr\bigg] \D^{1/2} \x
}
is convex in $\mu$.
\end{Theorem}

Theorem \ref{Theorem:ConvexMulti} has an important implication:  ``error catastrophes'' for individual genotypes as the mutation rate increases never correspond to a precipitous decline in population mean fitness.  The convexity of $r(\M(\mu)\D)$ in $\mu$ means that no such declines are \emph{ever} possible, that in contrast, $r(\M(\mu)\D)$ steadily declines at rates that never increase as $\mu$ increases.  These observation, noted by \citet{Hermisson:etal:2002:Mutation,Bull:AncelMeyers:Lachmann:2005,Tejero:Montero:and:Nuno:2015:Theories} based on specific fitness landscapes, are seen to hold for all possible fitness landscapes under reversible mutation.

\subsection{House-of-Cards Mutation}
The ``house-of-cards'' (HoC) mutation model introduced by \citet{Kingman:1977:Properties} is a single-event mutation model where $P_{ij} = \pi_i$ for all $j$, which makes $\P=\piv \ev\tr$ a rank-one matrix, and therefore $\lambda_i = 0$ for $i \geq 2$.
\begin{Result}[house-of-cards Mutation]
Let $\M(\mu) = (1-\mu) + \mu \P$, where $\P = \piv \ev\tr$.  Then
\ab{
r(\M(\mu) \D) & \leq (1{-}\mu)  \max_{i}[w_i] +  \mu \sumin \pi_i w_i .
}
\end{Result}
Under biological conditions the term $\sumin \pi_i w_i$ is near zero, since it reflects the average fitness of genotypes under the stationary distribution of mutation --- i.e.\ a genome made out of random nucleotides.  In this case, $\AMR(\M(\mu) \D) \approx 1-\mu$, which means that house-of-cards mutation is incapable of supporting any mutational robustness.  Note that HoC mutation at individual loci does not create HoC mutation for the entire genome.  But also note that HoC mutation at each locus combined with multiplicative non-epistasis ($\D = \Ox_{\xi=1}^L \D^{(\xi)}$) decomposes into multiple HoC systems.  Therefore, with multi-locus HoC mutation, multiplicative epistasis is required for population mutational robustness to rise above the Haldane limit.


\section{Discussion}

Priority here has been given to presenting the new mathematical results on these decades old problems rather than to their biological interpretations or to particular illustrations.  But a few words are in order.  The chief contribution is to show how the elevation of a population's mean fitness above the Haldane limit at a mutation-selection depends on underlying spectral relationships between the mutation structure and the array of fitnesses.   The results are obtained in almost complete generality, encompassing arbitrary fitnesses and for some results, arbitrary mutation structures, and in others, reversible multilocus mutation.  The results directly extend to a novel concept of ``dispersal robustness''.

Infinite population theory involving the spectral radius $r(\M\D)$ and Perron vector $\zvh > \0$ was originally developed in one or two locus theory, or dispersal models, where the population size could easily exceed the number of genotypes or patches, and large populations were well approximated by infinite population models.  In quasispecies theory, however, results were aimed to encompass the entire genome, entailing genotype spaces orders of magnitude larger than any possible population, rendering questionable the biological relevance of $\zvh$ and of $r(\M\D)$.

Their relevance may possibly be retained in several circumstances:  (1)  when a small number of loci contribute to a multiplicative fitness component, (2) when transient dynamics produce a meta-stable mutation-selection balance \citep{Nimwegen:and:Crutchfield:2000:Metastable}, (3) when the dynamics can be approximated by a coarse-graining to give a ``phenotypic quasispecies'' \citep{Nimwegen:and:Crutchfield:2001:Optimizing}, and (3) in a finite population multitype branching process.  In the latter, noted in \citet{Hermisson:etal:2002:Mutation}, the expected number of offspring $i$ produced from one parent $j$ defines a  matrix $\M\D$, and $(\M\D)^t \z(0)$ is the trajectory of expected numbers from an initial population, $\z(0)$.  The long-term expected growth rate of the population is then $r(\M\D)$. 

Each of these conditions merits further investigation.  With these considerations as caveats, let us review the results here.

The first main result (Theorem \ref{Theorem:Main:Eigenvalue}) is that the spectral gap of the mutation matrix sets an upper bound on how much above than the Haldane limit the equilibrium population fitness can be.  The spectral gap measures how \emph{rapidly mixing} the mutation operator is.   Theorem {Theorem:Main:Eigenvalue} thus shows that the more rapidly mixing mutation is, the small the mutational robustness at a mutation-selection balance.  In Kingman's well-known house-of-cards model, $\lambda_2=0$ so the spectral gap is nearly maximal, so almost no mutational robustness is possible.  Theorem \ref{Theorem:Main:Eigenvalue} provides an impetus to compare different kinds of mutation models, as well as dispersal models, for their spectral gaps and therefore how mutational robustness is limited.

The second main result (Theorem \ref{Theorem:Main:Eigenvector}) is that the alignment of the fitnesses with the orthogonalized eigenvectors of the mutation matrix, weighted by the eigenvalues, sets a lower bound on the population mutational robustness.  This bound is low when the landscape is rugged and fitnesses align with the higher-frequency eigenvectors, and high when the landscape is smooth and fitnesses align with the low-frequency eigenvectors.

Several additional results generalize results that have been found for special cases of fitness landscapes to arbitrary landscapes:  Theorems \ref{Theorem:SingleLoad} and \ref{Theorem:MultiLoad} show that the aggregate population mutational robustness can only decline with increasing mutation rate. Theorem \ref{Theorem:ConvexMulti} shows that this decrease happens at an at ever lessening rate as the mutation rate increases.  So ``error catastrophes'' for genotypes can never cause the population fitness to plummet \citep{Bull:AncelMeyers:Lachmann:2005}.  But high enough mutation rates will inescapably drive the population to the fitness of a random genotype sequence, so for some rate less than this, lethal mutagenesis is assured.  The issue of whether lethal mutagenesis can be evaded by the evolution of mutational robustness \citep{Tejero:Montero:and:Nuno:2015:Theories} is thus resolved:  with high enough mutation rates, it cannot.

These applications are to be seen as only initial illustrations of the potential usefulness of the analysis presented here.  Other population dynamical processes may similarly be clarified by viewing them through this lens of their eigenvalues an eigenvectors.

\section*{Acknowledgements} \acknowledgeText


\pagebreak
\centerline{\bf \LARGE Supporting Information}

\section*{Proofs of the Results}
For convenience, the results are restated along with their proofs.
\setcounter{Theorem}{1}
\begin{Result}[Haldane Limits on Mutational Robustness]
\ab{
\AMR(\M(\mu)\D) &\eqdef \frac{r(\M(\mu)\D)}{\max_i[w_i]}\geq (1 - \mu)^L.
}
\end{Result}
\begin{proof}
At equilibrium, we examine $\zh_1$ where $w_1 = \max_{i=1}^n w_i$.  Then, 
\ab{
\wbh \zh_1 &= M_{11} w_1 \zh_1 + \sum_{j=2}^n M_{1j} w_j \zh_j
}
which implies
\ab{
(\wbh-M_{11} w_1) \zh_1 &=   \sum_{j=2}^n M_{1j} w_j \zh_j \geq 0 \\
\stext{hence \citet[p. 62]{Altenberg:1984:Book} }
\AMR(\M(\mu),\D) = \frac{\wbh}{ w_1} & \geq  M_{11},
}
with equality if and only if for all $j=2, \ldots, n$ either $M_{1j}=0, w_j=0$, or $\zh_j=0$.
Further, $M_{11}(\mu) = \prod_{\xi=1}^L (1-\mu + \mu P^{(\xi)}_{11}),$ so $M_{11}(\mu) \geq (1- \mu)^L$, with equality if and only if $ P^{(\xi)}_{11} = 0$ for all loci $\xi$.
\end{proof}
See also \citet[eq. (2.3)]{Burger:and:Hofbauer:1994:Mutation}, 
\citet[pp. 149--150, eq. (5.31)--(5.34)]{Burger:2000}.

\setcounter{Theorem}{3}
\begin{Theorem}
Letting the average fitness of mutant offspring produced by parent $j$ be
\ab{
\wb^{\mut}_j = \frac{ \sum_{i \neq j}^n w_i M_{ij}}{ \sum_{i \neq j}^n M_{ij}},
}
the average mutant offspring fitness in a population at mutation-selection balance is therefore
\ab{
\wb^\mut (\M(\mu), \D) &\eqdef \sumjn \wb^{\mut}_j \frac{w_j \zh_j}{\wbh } \\&
= \wbh \left[1 - \frac{(1-\mu)^L}{1-(1-\mu)^L} \Var\big[\frac{w_i}{\wbh} \big]  \right],
}
assuming the mutation distributions have $P^{(\xi)}_{ii} {=} 0$ for all loci $\xi$ and alleles $i$. 
\end{Theorem}
\begin{proof}
The average mutant offspring offspring fitness from parent $j$, $\wb^{\mut}_j$, averaged over the population, censused after selection, before reproduction, is
\ab{
\wb^{\mut} &(\M(\mu), \D) = 
\sumjn \wb^{\mut}_j \frac{w_j \zh_j}{\wbh} 
= \sumjn \frac{\sum_{i \neq j}^n w_i M_{ij}}{ \sum_{i \neq j}^n M_{ij}} \frac{w_j \zh_j}{\wbh}\\&
= \sumjn \frac{\sumin w_i M_{ij} - w_j M_{jj}}{1- M_{jj}} \frac{w_j \zh_j}{\wbh}.
}
Recall that the multiple-event mutation rates are
\ab{
M_{jj} = \prod_{\xi=1}^L M^{(\xi)}_{j_\xi j_\xi}
=  \prod_{\xi=1}^L [(1-\mu) + \mu P^{(\xi)}_{j_\xi j_\xi}].
}
Assuming $ P^{(\xi)}_{j_\xi j_\xi} = 0$ for all $\xi$, $j$, then $M_{jj} = (1-\mu)^L$.  Substitution gives
\ab{
\wb^{\mut} (\M(\mu), \D)  
&= \sumjn \frac{\dsty \sumin w_i M_{ij} - w_j (1-\mu)^L}{1- (1-\mu)^L} \frac{w_j \zh_j}{\wbh} \\&
= \frac{\dsty \sumijn w_i M_{ij}\frac{w_j \zh_j}{\wbh} - \sumjn \frac{w_j^2 \zh_j}{\wbh} (1-\mu)^L}{1- (1-\mu)^L}.
}
Since
\ab{
\sumjn \frac{w_j^2 \zh_j}{\wbh} = {\wbh} \sumjn \frac{w_j^2 \zh_j}{\wbh^2}
= \wbh \left( 1 + \Var\big[\frac{w_j}{\wbh}\big] \right),
}
\ab{
\wb^{\mut} (\M(\mu), \D)  
&
= \frac{\dsty \sumin w_i \zh_i -\wbh \big( 1 + \Var\big[\frac{w_j}{\wbh}\big] \big) (1-\mu)^L}{1- (1-\mu)^L}  \\&
= \wbh - \wbh \left[\frac{ (1-\mu)^L}{1- (1-\mu)^L} \right]  \Var\big[\frac{w_j}{\wbh}\big] \\&
= \wbh \left[1 - \frac{(1-\mu)^L}{1-(1-\mu)^L} \Var\big[\frac{w_i}{\wbh} \big]  \right].
\qedhere
}
\end{proof}

\begin{Theorem}[Neutral Network Robustness]
The equilibrium population robustness $\nubh = r(\M\D) / w$ is always greater than the average robustness over the neutral network, 
\ab{
\Ex{\nu} 
= \frac{\sum_{i \in \Nc} \sum_{j \in \Nc} M_{ij}}{\sumin \sum_{j \in \Nc} M_{ij}}.
}
if 
\enumlist{
\item the mutation matrix $\M$ is symmetric, 
\item fitness off the network is $0$, and 
\item not all genotypes on the neutral network have the same mutational robustness.
}
\end{Theorem}
\begin{proof} 
Repeated use will be made of the identity $\sum_{i \in \Nc} f_i = \sumin  f_i w_i / w$ for arbitrary values $f_i$.  We may write the network's average robustness as
\ab{
\Ex{\nu} 
&= \frac{\sum_{i \in \Nc} \sum_{j \in \Nc} M_{ij}}{\sumin \sum_{j \in \Nc} M_{ij}}
= \frac{1}{w} \frac{\sumin w_i \sumjn M_{ij} w_j}{\sumin \sumjn M_{ij} w_j} \\&
= \frac{1}{w} \frac{\evt \D\M \D \ev }{ \evt \D \ev}. 
}
We have $r(\D\M\D) =  w \, r(\M\D)$ since $(\D\M\D)^t = w^t \D(\M\D)^t = w^t (\D\M)^t \D$.

Since $w_i \in \{w, 0\}$, $\D^2 = w \D $, so $\sqrt{w} \D^{1/2} =  \D$, hence
\an{
\Ex{\nu} &=\frac{1}{w} \frac{\evt \D\M \D\ev }{ \evt \D \ev} 
= 
\frac{1}{w^2} \frac{(\evt \D) \D\M \D (\D \ev)}{ (\evt \D) (\D \ev) } \notag\\&
\leq \frac{1}{w^2} \max_{\x\neq \0}\frac{\x\tr \D \M \D \x} { \x \tr \x } \label{eq:eDMDe}\\&
= \frac{1}{w^2} \  r(\D\M \D) 
 =  \frac{r(\M \D)}{w} = \AMR(\M,\D) . \notag
}
Equality holds if and only if $\ev\tr \D$ is a left eigenvector of $\D\M\D$ \citep[Theorem 4.2.2 (Rayleigh)]{Horn:and:Johnson:2013}, in which case 
\ab{
 \lambda \ev\tr \D = (\ev\tr \D) \D \M \D = w \ev\tr \D \M \D,
}
and so $\ev\tr \D$ is also an eigenvector of $\M\D$ with eigenvalue $\lambda/w$.  This is equivalent to the condition that all genotypes on the network have the same neutrality, i.e. for all $j \in \Nc$, $c= \sum_{i \in \Nc} M_{ij} = \sumin \frac{w_i}{w} M_{ij}$, or 
\ab{
c \, \ev\tr\D = \frac{1}{w^2} \ev\tr \D \M \D
}
with  $\lambda  = c / w$.  Therefore, if there are any differences in the neutralities on the network, the inequality \eqref {eq:eDMDe} is strict.
\end{proof}

\setcounter{Theorem}{6}
\begin{Theorem}[Corollaries 1, 3 of \citet {Altenberg:2011:Mutation}]
The equilibrium genetic load of a quasispecies under irreducible multiple-event mutation strictly increases with mutation rate $0 < \mu < 1/2$, when the quasispecies has variation in fitness within its irreducible blocks.
\end{Theorem}
\begin{proof}
Corollaries 1, 3 of \citet {Altenberg:2011:Mutation} show that 
\an{
r(\M(\mu)\D) = r\big(\Ox_{\xi=1}^L[(1-\mu) \I^{(\xi)} + \mu \P^{(\xi)}] \D \big) \label{rbigOx}
}
strictly decreases in $\mu \in (0, 1/2)$ when $\D \neq c \ \I$ for any $c > 0$, and when there is only one equilibrium for each $\mu$.  In the case where there are lethal genotypes, there may be more than one equilibrium, corresponding to multiple quasispecies.  In that case, one must focus on the irreducible blocks, $\kappa$, from the Frobenius normal form of the reducible matrix $\M(\mu)\D$.  The irreducible restriction $r([\M(\mu)\D]^{[\kappa]})$ has a unique equilibrium for each $\mu$, and Corollaries 1 and 3 apply, meaning that if there is any variation in fitness among $\{w_i\suchthat i \in \kappa\}$, then $r([\M(\mu)\D]^{[\kappa]})$ decreases strictly in $\mu$.  The quasispecies may comprise more than one block, but its mean fitness, $\max_\kappa \{r([\M(\mu)\D]^{[\kappa]})\}$, will therefore also be decreasing in $\mu$. 
\end{proof}

\setcounter{Theorem}{8}
\begin{Result}[Single-Event Relative Robustness Increases with Mutation Rates] 
As the mutation rate increases with single-event mutation, if a quasispecies has more than one [only one] nonlethal fitness value, then the relative mutational robustness $\RMR(\M(\mu), \D)$ strictly increases [is constant] with mutation rate. 
\end{Result}
\begin{proof}
We examine the restriction of $\M(\mu)\D$ to an irreducible block $\M(\mu)^{[\kappa]} \D^{[\kappa]}$.  In the neutral network case where $w_i = w$ for all $i \in \kappa$, then
\ab{
r&(\M(\mu)^{[\kappa]} \D^{[\kappa]}) = w \, r([(1-\mu) \I^{[\kappa]} + \mu \P^{[\kappa]})] \\&
= w [(1-\mu) + \mu r(\P^{[\kappa]})]
= w [1 + \mu (r(\P^{[\kappa]}) - 1)],
}
so with single-event mutation, $L=1$, and
\ab{
\RMR(\M(\mu)^{[\kappa]}, \D^{[\kappa]}) 
&= \frac{r(\M(\mu)^{[\kappa]} \D^{[\kappa]})/w - 1 + \mu}{1 - 1 + \mu} \\&
= \frac{1 + \mu (r(\P^{[\kappa]}) - 1) - 1 + \mu}{ \mu} \\&
= r(\P^{[\kappa]}) < 1.
}
which shows the relative robustness is a constant for all $\mu$.

For brevity write $f(\mu) \eqdef r(\M(\mu)^{[\kappa]} \D^{[\kappa]}) / \max_i[w_i]$ and 
\ab{
g(\mu) &\eqdef \RMR (\M(\mu)^{[\kappa]}, \D^{[\kappa]}) 
= \frac{ f(\mu) - 1+\mu}{\mu} \\&
= 1 +  \frac{ f(\mu) - 1}{\mu}.
}
Then
\ab{
g(\mu_1) - g(\mu_2) &
= \frac{ f(\mu_1) - 1}{\mu_1} - \frac{ f(\mu_2) - 1}{\mu_2}.
}
Let $\mu_1 = \mu > 0$ and $\mu_2 = h \mu$, $0<h<1$.  
\an{
g(\mu) - g(h \mu) &
= \frac{h f(\mu) - h - (f(h \mu) - 1)}{\mu h}. \label{eq:gg}
}

In the case where there are more than one nonlethal fitness in the irreducible block, $ \D^{[\kappa]}$ is nonscalar ($ \D^{[\kappa]} \neq c \  \I^{[\kappa]}$ for any $c \in \Reals$).  In \citet{Altenberg:2012:Resolvent-Positive} it is shown, by applying ``dual convexity'' to theorems Cohen \citep{Cohen:1981:Convexity} and Friedland \citep{Friedland:1981}, that $r([(1-\mu) \I^{[\kappa]} + \mu \P^{[\kappa]}] \D^{[\kappa]}) $ is strictly convex in $\mu$ when $\D^{[\kappa]}$ is nonscalar.  Convexity means that for all $0 < p < 1$, and $\alpha_1, \alpha_2$,
\an{
(1-p) f(\alpha_1) + p f(\alpha_2) > f((1-p) \alpha_1 + p \alpha_2). \label{eq:1pf}
}
Let $\alpha_1 = 0$ and $\alpha_2 = \mu$.  We know that 
\ab{
f(0) &= \frac{r(\D^{[\kappa]})} {\max_{i=1,\ldots,n} w_i }
= \frac{\max_{i \in \kappa} w_i}{\max_{i=1,\ldots,n} w_i } \leq 1.
}
So \ref{eq:1pf} becomes
\ab{
1-p + p f(\mu) \geq (1-p) f(0) + p f(\mu) > f(p \mu). 
}
Now, let $p=h$.  Using $f(0) \leq 1$ and \eqref{eq:gg} we get
\ab{
0 &< (1-h) f(0) + h f(\mu)  - f(h \mu) \\&
\leq 1-h  + h f(\mu)  - f(h \mu) 
= \mu h( g(\mu) - g(h \mu) ).
}
Therefore $g(\mu) = \RMR(\M(\mu)^{[\kappa]}, \D^{[\kappa]})$ strictly increases in $0<\mu<1$ for any quasispecies $\kappa$.
\end{proof}

\begin{Theorem}[Expression for the Spectral Radius]
When mutation has the transition matrix of a reversible Markov chain, $\M = \D_{\piv}^{1/2}\K \Lam \K\tr \D_{\piv}^{-1/2}$, then
\an{
r(\M\D) &= \max_{\x\tr\x = 1}  \x\tr \D^{1/2}\K \Lam \K\tr  \D^{1/2} \x \notag \\&
= \max_{\x\tr\x = 1} \sumjn \lambda_j \bigg(\sumin x_i \sqrt{w_i} K_{ij}\bigg)^2. \label{eq:rMD2}
}
\end{Theorem}
\begin{proof}
The transition matrix of an irreducible reversible Markov chain can be represented (\citet[p. 33]{Keilson:1979} \cite{Altenberg:2011:Mutation}) as
\ab{
\M = \D_{\piv}^{1/2}\K \Lam \K\tr \D_{\piv}^{-1/2}
}
where reversibility requires $M_{ij} \pi_j = M_{ji} \pi_i$, and
\desclist{
\item[$\piv$] $ = (\pi_1, \ldots, \pi_n)\tr = \M \piv$ is the stationary distribution for $\M$, i.e.\ its \emph{Perron vector},
\item[$\D_{\piv}^{1/2}$] is the diagonal matrix of the square roots of $\pi_i$,
\item[$\K$] is the matrix of orthogonalized eigenvectors of $\M$, $\K \K\tr = \K\tr\K = \I$,  i.e.\  for each $j=1, \ldots, n$, $\sum_{i=1}^n K_{ij}^2 = 1$, $\sum_{i=1}^n K_{ij} K_{ih} = 0$ if $j\neq h$, 
\item[\protect{$[\K]_j$}] is the $j$-th column of $\K$, 
\item[$K_{i1} = \pi_i^{1/2}$] or $[\K]_1 = \D_\piv^{1/2} \ev$,
\item[$\Lam$] $ \eqdef \diag{\lambda_i}$ is the diagonal matrix of eigenvalues  $\lambda_1 > \lambda_2 \geq \cdots \geq \lambda_n \geq -1$ of $\M$, these inequalities coming from Perron-Frobenius theory.
}

Applying this representation of $\M$, we find that the spectral radius of $\M \D$ is that of a symmetric matrix:
\ab{
r(\M \D) &
= r(\D_{\piv}^{1/2}\K \Lam \K\tr \D_{\piv}^{-1/2} \D) \\&
= r(\K \Lam \K\tr \D_{\piv}^{-1/2} \D \D_{\piv}^{1/2}) 
= r(\K \Lam \K\tr \D) \\&
= r( \Lam \K\tr \D \K) 
= r(\D^{1/2}\K \Lam \K\tr  \D^{1/2}).
}
Note the remarkable symmetry here: $r(\M\D) =r(\K \Lam \K\tr \D) = r( \Lam \K\tr \D \K)$, which sandwiches one of two diagonal matrices, $\Lam$ and $\D$, between the change-of-basis matrices $\K$ and $\K\tr$.

From this symmetric form, the theorem results immediately from application of the Rayleigh quotient characterization of the spectral radius of any symmetric real matrix $\Sm$ (\citealt[pp. 172--173]{Wilkinson:1965}, \citealt[pp. 176--180]{Horn:and:Johnson:1985}),
$ 
r(\Sm) = \max_{\x \neq \0}({\x\tr \Sm \x}/{\x\tr\x}). \qedhere
$ 
\end{proof}
{\bf Remark 1.}  No assumptions on the irreducibility of $\M$ or $\M\D$ enter here.  In order for $r(\M\D)$ to be the mean fitness of the population when $\M\D$ is reducible, we must impose the additional assumption that the population contain a quasispecies $\kappa$ such that $r(\M^{[\kappa]} \D^{[\kappa]}) = r(\M\D)$.

{\bf Remark 2.}  The approach taken here may have utility for graph theory.  One obtains bounds on the spectral radius of a graph by embedding it into a larger graph whose eigenvectors and eigenvalues are known, and projecting those eigenvectors onto the graph.

\setcounter{Theorem}{11}
\begin{Theorem}{\bf (Lower Bound from the Fourier Representation of Fitness)}
Represent the fitnesses $w_i$ as $w_i =  \sumjn K_{ij} a_j$ or $\wv = \K \av$, where $a_i$ is the Fourier coefficient for eigenvector $[\K]_i$.  Then 
\an{\label{eq:la2}
r(\M\D) & \geq \frac{1}{\sumjn w_j} \sumjn \lambda_j a_j^2 .
}
\end{Theorem}
\begin{proof}
The square-root terms $\sqrt{w_i}$ in \eqref {eq:rMD2} can be ``digested'' by employing $x_i = c \sqrt{w_i}$, where the constraint $\x\tr\x=1$ implies $c = 1 / {\sqrt{\ev\tr\D\ev}}.$  Then
\an{
r(\M\D) &
= \max_{\x\tr\x = 1} \sumjn \lambda_j \left(\sumin x_i \sqrt{w_i} K_{ij}\right)^2 \notag\\&
\geq c^2 \sumjn \lambda_j \left(\sumin \sqrt{w_i} \sqrt{w_i} K_{ij}\right)^2 \notag\\&
= c^2 \sumjn \lambda_j \left(\sumin w_i K_{ij}\right)^2 \notag\\&
= \frac{1}{{\ev\tr\D\ev}} \ \ev\tr \D \K \Lam \K\tr \D \ev. \label{eq:DKLKD}
}

Substitution with the Fourier representation $\wv \equiv \D \ev = \K \av$ into \eqref{eq:DKLKD} gives \eqref {eq:la2}:
\ab{
r(\M\D)& \geq \frac{1}{{\ev\tr\D\ev}}  \ \ev\tr \D \K \Lam \K\tr \D \ev
\\&
= \frac{1}{{\sumjn w_j}} \  \av\tr \K\tr\K \Lam \K\tr  \K \av \\&
=
\frac{1}{{\sumjn w_j}} \  \av\tr  \Lam  \av
= \frac{1}{{\sumjn w_j}} \sumjn \lambda_j a_j^2 . \qedhere
}
\end{proof}

Several lemmas are needed for the proofs to follow.  The first is a generalization of the result for graphs in \citet{Stadler:1996:Landscapes}.
\begin{Lemma}[Parent-Offspring Fitness Covariance]\label{Lemma:Cov}
For a stationary mutational random walk with reversible Markov chain transition matrix $\M$, and stationary distribution $\piv$, the $\piv$-weighted covariance between parent and offspring fitnesses ($W_\Pc$, $W_\Oc$ respectively) is 
\an{
&\Cov_{\piv}[W_{\Pc}, W_{\Oc}]  = \sum_{j=2}^n \lambda_j \left(\sumin w_i \pi_i^{1/2} K_{ij}\right)^2.  \label{eq:Cov}
}
\end{Lemma}
\begin{proof}
\an{
&\Cov_{\piv}[W_{\Pc}, W_{\Oc}]  \notag\\&
= \sumjn \pi_j w_j \sum_{i=1}^n w_i M_{ij} - 
\bigg(\sumjn \pi_j  w_j \bigg)
\bigg(\sum_ {i,j=1} ^n \pi_j w_i M_{ij}  \bigg) \notag\\&
= \w\tr \M \D \piv - (\w\tr\piv)(\w\tr \M \piv) \notag\\&
= \w\tr (\D_\piv^{1/2} \K \Lam \K\tr \D_\piv^{-1/2}) \D \piv - \w\tr\piv\w\tr\piv  \label{eq:CovDeriv}\\&
= \ev\tr \D \D_\piv^{1/2} \K \Lam \K\tr \D_\piv^{1/2} \D  \ev -  \ev\tr \D \D_\piv \ev \ev\tr \D \D_\piv  \ev \notag\\&
= \ev\tr \D \D_\piv^{1/2} (\K \Lam \K\tr) \D_\piv^{1/2} \D  \ev  \notag\\&
\qquad -  \ev\tr \D \D_\piv^{1/2}( \D_\piv^{1/2}\ev \ev\tr  \D_\piv^{1/2} ) \D_\piv^{1/2} \D \ev \notag\\&
= \ev\tr \D \D_\piv^{1/2} (\K \Lam \K\tr -  \D_\piv^{1/2}\ev \ev\tr  \D_\piv^{1/2} ) \D_\piv^{1/2} \D  \ev \notag \\&
= \ev\tr \D \D_\piv^{1/2} (\K \Lam \K\tr -  [\K]_1 [\K_1]\tr) \D_\piv^{1/2} \D  \ev \notag \\&
= \ev\tr \D \D_\piv^{1/2} (\K \diag{0, \lambda_2, \ldots, \lambda_n} \K\tr ) \D_\piv^{1/2} \D  \ev  \notag \\&
= \sum_{j=2}^n \lambda_j \left(\sumin w_i \pi_i^{1/2} K_{ij}\right)^2. \label{eq:CovSum}
}
\end{proof}

The expression \eqref{eq:Cov} simplifies when $\M$ is symmetric, to give the following, originally derived in \citet{Stadler:1996:Landscapes,Reidys:and:Stadler:2001:Neutrality}.
\begin{Lemma}[Parent-Offspring Fitness Covariance Under Symmetric Mutation \citet{Stadler:1996:Landscapes}]
For a stationary mutational random walk with a doubly stochastic transition matrix $\M=\M\tr= \K \Lam \K\tr$, and stationary distribution $\piv=\ev / n$, the parent-offspring covariance in fitness can be represented in terms of the Fourier representation of fitnesses, $\w = \K \av$, as
\ab{
\Cov_{\ev / n}&[W_{\Pc}, W_{\Oc}] =  \frac{1}{n} \sum_{i=2}^n \lambda_i a_i^2 .
}
\end{Lemma}
\begin{proof}
When $\M$ is symmetric, $\piv = \ev / n$, $[\K]_1 = \ev / \sqrt{n}$, and $\ev\tr \K = \sqrt{n} \ \ev_1\tr$, and using the Fourier representation $\w = \K \av$, one gets
\ab{
\sumjn w_j = \ev\tr \w = \ev\tr \K \av = \ev_1 \tr \av \sqrt{n}= a_1 \sqrt{n}.
}
Substituting the above into \eqref{eq:CovDeriv}, one obtains:
\ab{
\Cov_{\ev / n}[&W_{\Pc}, W_{\Oc}] 
= \w\tr (\D_\piv^{1/2} \K \Lam \K\tr \D_\piv^{1/2}) \w {-} (\w\tr\piv)^2 \\ &
= \w\tr ( \K \Lam \K\tr ) \w / \sqrt{n^2} - (\w\tr\ev)^2 / n^2 \\ &
= \av\tr \K\tr \K \Lam \K\tr \K \av / n - (\av\tr \K\tr \ev)^2 / n^2 \\ &
= \av\tr \Lam \av / n - (\av\tr \ev_1 \sqrt{n})^2 / n^2 \\ &
=  \frac{1}{n} \left(\sumin \lambda_i a_i^2 - a_1^2 \right)
=  \frac{1}{n} \sum_{i=2}^n \lambda_i a_i^2.& \qedhere
}
\end{proof}
\begin{Theorem}[Relation to Random Mutational Walks]
For a stationary mutational random walk with reversible Markov chain transition matrix $\M$, and stationary distribution $\piv$, let $\Cov_\piv(W_{\Pc}, W_{\Oc})$ be the $\piv$-weighted covariance between parent and offspring fitnesses.  Then at a mutation-selection balance, $\wbh \zvh = \M\D\zvh$,
\ab{
r(\M\D) \geq \left(\sumjn w_j \pi_j \right) \left[1 + \frac{1}{{(\w\tr\piv)}^2} \Cov_\piv[ {W_\Pc}, {W_\Oc}]\right].
}
\end{Theorem}
\begin{proof}
Using \eqref {eq:CovDeriv},
\ab{
r&(\M\D) = \max_{\x \neq \0} \frac{\x\tr (\D^{1/2} \K \Lam \K\tr \D^{1/2}) \x}{\x\tr\x} \\&
\geq \frac{\ev\tr (\D_\piv \D)^{1/2}( \D^{1/2} \K \Lam \K\tr \D^{1/2}) (\D_\piv \D)^{1/2} \ev}{\ev\tr(\D_\piv \D)\ev} \\&
= \frac{\ev\tr  \D \D_\piv^{1/2} \K \Lam \K\tr \D_\piv^{1/2} \D \ev}{\w\tr\piv} \\&
= \frac{1} {\w\tr\piv} [\Cov_{\piv}[W_{\Pc}, W_{\Oc}] + (\w\tr\piv)^2] \\&
= (\w\tr\piv) \left[ 1 + \frac{1} {(\w\tr\piv)^2} \Cov_{\piv}[W_{\Pc}, W_{\Oc}] \right]. \qedhere
}
\end{proof}

\begin{Corollary}[Condition for Positive Parent/Offspring Fitness Covariance]
When the eigenvalues of $\M$ are positive, then $\Cov_\piv[ {W_\Pc}, {W_\Oc}] \geq 0$.
\end{Corollary}
\begin{proof}
When $\lambda_i > 0$ for all $i$, then all terms in the sum \eqref{eq:Cov} are nonnegative.
\end{proof}

\begin{Lemma}\label{Lemma:Kbasis}
Since $\K$ is orthogonal, its columns form a basis to represent $\x = \K \cv$, or $x_i = \sumjn K_{ij} c_j$.  The constraint $\x\tr \x = 1 = \cv\tr \K\tr \K \cv = \cv\tr\cv$ implies $\cv\tr\cv=1$.  Then
\ab{
\sum_{i =1}^n x_i K_{ij} &= \sum_{i =1}^n (\sum_{h=1}^n K_{ih} c_h) K_{ij} \\&
= \sum_{h=1}^n c_h \sum_{i =1}^n  K_{ih}  K_{ij}
= \sum_{h=1}^n c_h  \delta_{hj} = c_j.
}
\end{Lemma}

\begin{Lemma}\label{Lemma:xmaxy}
\an{
\max_{\x\tr\x=1} \x\tr \y = (\y\tr\y)^{1/2}, \ \min_{\x\tr\x=1} \x\tr \y = -(\y\tr\y)^{1/2}.\label{eq:xmaxy}
}
\end{Lemma}
\begin{proof}
Critical points of
$
\max_{\x\tr\x=1}\x\tr \y 
= \max_{\x\neq \0} \frac{ \x\tr \y}{\x\tr\x}
$ 
satisfy
\ab{
\pf{}{x_i}\frac{ \x\tr \y}{\x\tr\x} &= \frac{y_i}{\x\tr\x} - x_i \frac{ \x\tr \y}{(\x\tr\x)^2}
= 0,
}
which, for $\x\tr\x=1$, are solved by $y_i = x_i (\x\tr\y)$,  hence $\y\tr\y = (\x\tr\y)^2$ so $\sqrt{\y\tr\y} = |\x\tr\y|$ is the maximum, and $-\sqrt{\y\tr\y} = -|\x\tr\y|$ the minimum.
\end{proof}

\setcounter{Theorem}{14}
\begin{Theorem}[Upper Bound from the Spectral Gap] 
Assuming $0 < \mu < 1/2$, then
\ab{
\sumin w_i \pi_i 
&\leq r(\M\D) \\&
\leq \sumin w_i \pi_i + \lambda_2(\M) \left(\max_i[w_i]- \sumin w_i \pi_i\right) ,
}
where $\lambda_2(\M) = 1{-}\mu {+} \mu {\dsty \max_{\xi=1, \ldots, L}} \lambda_{2}(\P^{(\xi)})$.
\end{Theorem}

\begin{proof}
The condition $0 < \mu < 1/2$ assures that all the eigenvalues of $\M(\mu)^{(\xi)}$ are positive, and thus all $\lambda_i(\M(\mu)) > 0$.  A series of inequalities are obtained, the last steps using the representation $\x = \K \c$ from Lemma \ref {Lemma:Kbasis}, and Lemma \ref {Lemma:xmaxy}.
\ab{
&r (\M\D) = \\&
\max_{\x\tr\x = 1}\bigg[( \x\tr \D^{1/2} \piv^{1/2})^2 + \sum_{j=2}^n \lambda_j ( \x\tr \D^{1/2} [\K]_j)^2 \bigg] \\&
 \leq 
\max_{\x\tr\x = 1}\bigg[( \x\tr \D^{1/2} \piv^{1/2})^2 + \lambda_2 \sum_{j=2}^n ( \x\tr \D^{1/2} [\K]_j)^2 \bigg] \\&
= \max_{\x\tr\x = 1}\bigg[(1{-}\lambda_2) ( \x\tr \D^{1/2} \piv^{1/2})^2 \ +\\&
\qquad  \lambda_2 \bigg(( \x\tr \D^{1/2} \piv^{1/2})^2+ \sum_{j=2}^n ( \x\tr \D^{1/2} [\K]_j)^2 \bigg) \bigg] 
}
\ab{&
= \max_{\x\tr\x = 1}\! \! \bigg[(1{-}\lambda_2) ( \x\tr \D^{1/2} \piv^{1/2})^2 
{+} \lambda_2 \sumjn ( \x\tr \D^{1/2} [\K]_j)^2 \bigg] \\&
\leq  (1{-}\lambda_2)\max_{\x\tr\x = 1} ( \x\tr \D^{1/2} \piv^{1/2})^2 \ + \\&
\qquad \qquad \lambda_2 \max_{\x\tr\x = 1} \sumjn ( \x\tr \D^{1/2} [\K]_j)^2  \\&
\leq  (1{-}\lambda_2) \sumin w_i \pi_i + 
\lambda_2 \max_i[w_i] \max_{\x\tr\x = 1} \sumjn ( \x\tr [\K]_j)^2  \\&
=  (1{-}\lambda_2) \sumin w_i \pi_i + 
\lambda_2 \max_i[w_i] \max_{\c\tr\c = 1} \sumjn c_j^2  \\&
=  (1{-}\lambda_2) \sumin w_i \pi_i + 
\lambda_2 \max_i[w_i].
}

The lower bound, $r(\M\D) \geq \sumin w_i \pi_i$,  was proven in \citet[Corollary F.2.]{Karlin:1982}.  Here, a separate proof is provided.
\ab{
r(&\M\D) \\&
 = \max_{\x\tr\x = 1}\bigg[( \x\tr \D^{1/2} \piv^{1/2})^2 
+ \sum_{j=2}^n \lambda_j  (\sumin x_i \sqrt{w_i} K_{ij})^2 \bigg] \\&
\geq \max_{\x\tr\x = 1}\bigg[( \x\tr \D^{1/2} \piv^{1/2})^2 
+ \lambda_n \sum_{j=2}^n   (\sumin x_i \sqrt{w_i} K_{ij})^2 \bigg] \\
& \geq \max_{\x\tr\x = 1}\bigg[( \x\tr \D^{1/2} \piv^{1/2})^2  \bigg] = \sumin w_i \pi_i .& \qedhere 
} 
\end{proof}

\ \\ 
\setcounter{Theorem}{16}
\begin{Theorem}[Convexity of the Spectral Radius in $\mu$]
In the case where $L$ loci mutate independently at rate $\mu$, each in a reversible Markov chain, then for $0 < \mu <  1/2$, $r(\M(\mu) \D)  =$
\ab{
\dsty &\max_{\x\tr\x = 1} 
\x\tr \D^{1/2} \bigg[\! \Ox_{\xi=1}^L \K^{(\xi)} [(1{-}\mu) \I^{(\xi)} {+} \mu \Lam^{(\xi)}]  {\K^{(\xi)}}\tr\bigg] \D^{1/2} \x
}
is convex in $\mu$.
\end{Theorem}

\begin{proof}
For the situation where $L$ loci mutate independently at rate $\mu$, and locus $\xi$ has mutation distribution matrix $\P^{(\xi)}$, the mutation matrix for the entire genome is 
\ab{
\M(\mu)
&\!= \!\Ox_{\xi=1}^L  [(1{-}\mu) \I^{(\xi)} {+} \mu \P^{(\xi)}]
= \!\Ox_{\xi=1}^L  [ \I^{(\xi)} {+} \mu (\P^{(\xi)} {-}  \I^{(\xi)})].
}
For the case where mutation at each locus forms a reversible Markov chain,
\ab{
\P^{(\xi)} = \D_{\piv_\xi}^{1/2} \K^{(\xi)}\Lam^{(\xi)}{\K^{(\xi)}}\tr \D_{\piv_\xi}^{-1/2}.
}
Then by the Rayleigh theorem,
\an{
&r(\M(\mu) \D)=   \label{eq:rMulti} \\&
\max_{\x\tr\x = 1} \!
\x\tr \D^{1/2} \bigg[\!\Ox_{\xi=1}^L \K^{(\xi)} [(1{-}\mu) \I^{(\xi)} {+} \mu \Lam^{(\xi)}]  {\K^{(\xi)}}\tr\bigg] \D^{1/2} \x. \notag 
}
The total number of genotypes is $|G| \eqdef \prod_{\xi=1}^L n_\xi$, where $n_\xi$ is the number of alleles at the $\xi$-th locus.  As a way to index the genotypes for summation we can use lexicographical order, 
\ab{
\Lcu \suchthat \{1, \ldots, n_1\} \times \cdots \times \{1, \ldots, n_L\} \goesto \{1, \ldots, |G|\},
}
where we use the notation $j = \Lcu(j_1, j_2, \ldots, j_L) \in \{1, \ldots, \prod_{\xi=1}^L n_\xi\}$, and $j_\xi \eqdef [\Lcu^{-1}(j)]_\xi$.

Now \eqref{eq:rMulti} can be expressed as
\ab{
r(\M(\mu)& \D) \\
&= \max_{\x\tr\x = 1} \sum_{j=1}^{|G|} \bigg[\prod_{\xi=1}^L (1{-}\mu {+} \mu \lambda^{(\xi)}_{j_\xi}) \bigg](\x\tr \D^{1/2} [\K]_j)^2,
}
where $[\K]_j$ is the $j$-th column of $\K = \OxL \K^{(\xi)}$.  Define 
\ab{
\beta_j(\mu) = \prod_{\xi=1}^L \left(1-\mu + \mu \lambda^{(\xi)}_{j_\xi}\right)
= \prod_{\xi=1}^L \left[1{-}\mu(1-\lambda^{(\xi)}_{j_{\xi}}) \right]  .
}

First we will see that each $\beta_j(\mu)$ is convex in $\mu$.   The case $L=1$ is subsumed in Result \ref{Result:Increases}.  We assume $L \geq 2$.   The first derivative is
\ab{
\df{}{\mu} \beta_j(\mu)
&= 
\sum_{\kappa=1}^L -(1-\lambda^{(\kappa)}_{j_ \kappa}) \prod_{\xi \neq \kappa} [1{-}\mu(1-\lambda^{(\xi)}_{j_{\xi}}) ] .
}
The second derivative is
\an{
&\ddf{}{\mu} \beta_j(\mu) \notag \\
&= 
\sum_{\kappa=1}^L -(1-\lambda^{(\kappa)}_{j_ \kappa}) 
\sum_{\gamma \neq \kappa} -(1-\lambda^{(\gamma)}_{j_ \gamma}) \prod_{\xi \neq \kappa, \gamma} [1{-}\mu(1-\lambda^{(\xi)}_{j_{\xi}}) ] \notag \\&
= 
\sum_{\kappa=1}^L \sum_{\gamma \neq \kappa} 
(1-\lambda^{(\kappa)}_{j_ \kappa}) 
(1-\lambda^{(\gamma)}_{j_ \gamma}) 
\prod_{\xi \neq \kappa, \gamma} [1{-}\mu(1-\lambda^{(\xi)}_{j_{\xi}}) ], \label{eq:2nd}
}
where for $L = 2$ there is no term $\prod_{\xi \neq \kappa, \gamma} [1{-}\mu(1-\lambda^{(\xi)}_{j_{\xi}}) ] $.

All of the terms $(1-\lambda^{(\kappa)}_{j_ \kappa})$, $(1-\lambda^{(\gamma)}_{j_ \gamma})$ are nonnegative.
To have $1-\mu(1-\lambda^{(\xi)}_{j_{\xi}})  > 0$ for all $j_{\xi}$, the upper bound on $\mu$ is 
\an{\label{eq:mu}
\mu^* \eqdef  \min_{j_{\xi}} \frac{1}{1-\lambda^{(\xi)}_{j_{\xi}}} = \frac{1}{1-\min_{j_{\xi}} \lambda^{(\xi)}_{j_{\xi}}} \geq 1/2 ,
}
since for irreducible $\P^{(\kappa)}$ with real eigenvalues, Perron-Frobenius theory gives $1 \geq \lambda^{(\kappa)}_{j_ \kappa} \geq -1$.

It takes only one positive term in the sum to make $\ddf{}{\mu} \beta_j(\mu)$ positive.  The only way to get all the terms to be zero is for $j_\kappa =1$ or $j_\gamma= 1$ for every distinct pair of loci $\{\kappa, \gamma \} \subset \{1, \ldots, L\}$, which requires that index $j$ have only one $j_\xi \neq 1$.  For all the $j$ with at least two loci where $j_\gamma \neq 1$ and $j_\kappa \neq 1$, the term $\beta_j(\mu)$ is strictly convex in $\mu$.  For the other $j$ where $j_\xi= 1$ except for at most one locus, $\beta_j(\mu)$ is a straight line.  From this convexity we know that for any $j \in \{1, \ldots, |G|\}$, for $h \in (0,1)$,
\ab{
(1-h) \beta_j(\mu_1) + h \beta_j(\mu_2) \geq \beta_j( (1-h) \mu_1 + h \mu_2).
}

Now we return to the full expression,
\ab{
r(\M(\mu) \D)  
&= \max_{\x\tr\x = 1} \sum_{j=1}^{|G|} \beta_j(\mu) (\x\tr \D^{1/2} [\K]_j)^2.
}
For a given $\x$, each term $(\x\tr \D^{1/2} [\K]_j)^2$ is nonnegative, so
\ab{
\sum_{j=1}^{|G|} &\bigg[ (1-h) \beta_j(\mu_1) + h \beta_j(\mu_2) \bigg]  (\x\tr \D^{1/2} [\K]_j)^2 \\ &
= (1{-}h) \sum_{j=1}^{|G|} \beta_j(\mu_1) (\x\tr \D^{1/2} [\K]_j)^2 \\
& \qquad + h  \sum_{j=1}^{|G|} \beta_j(\mu_2) (\x\tr \D^{1/2} [\K]_j)^2 \\&
\geq \sum_{j=1}^{|G|} \bigg[ \beta_j( (1-h) \mu_1 + h \mu_2) \bigg] (\x\tr \D^{1/2} [\K]_j)^2 .
}
Hence, for any $\mu_1, \mu_2 \in (0, 1/2)$, and $h \in (0,1)$, 
\ab{
(1-h) \, r(\M & (\mu_1) \D) + h \, r(\M(\mu_2) \D) \\
= 
(1{-}h) \, & \max_{\x\tr\x=1} \sum_{j=1}^{|G|} \beta_j(\mu_1) (\x\tr \D^{1/2} [\K]_j)^2 \\
+ \ h  & \max_{\x\tr\x=1} \sum_{j=1}^{|G|} \beta_j(\mu_2) (\x\tr \D^{1/2} [\K]_j)^2 \\
\geq \max_{\x\tr\x=1}& \sum_{j=1}^{|G|}\beta_j( (1-h) \mu_1 + h \mu_2) (\x\tr \D^{1/2} [\K]_j)^2 \\
= r(\M( (&1-h) \mu_1 + h \mu_2) \D).
}
Therefore, $r(\M(\mu)\D)$ is convex in $\mu$.  With additional care, the conditions for strict convexity can be elicited.
\end{proof}

\begin{Result}[House of Cards Mutation]\label{Result:HoC}
Let $\M(\mu) = (1-\mu) \I + \mu \P$, where $\P = \piv \ev\tr$.  Then
\ab{
r(\M(\mu) \D) & \leq (1{-}\mu)  \max_{i}[w_i] +  \mu \sumin \pi_i w_i .
}
\end{Result}
\begin{proof}
\ab{
r(\M&(\mu)\D) = r([(1-\mu) \D + \mu \piv \ev\tr \D) \\&
= r([(1-\mu) \D + \mu  \D^{1/2}\D_{\piv}^{1/2} \ev \ev\tr \D_{\piv}^{1/2} \D^{1/2}) \\&
= \max_{\x\tr\x = 1} \x\tr [ (1{-}\mu) \D   +  \mu  \D^{1/2} \D_{\piv}^{1/2} \ev \ev\tr \D_{\piv}^{1/2}  \D^{1/2} ]\x \\&
= \max_{\x\tr\x = 1} \left[ (1{-}\mu)  \sumin w_i x_i^2  +  \mu ( \x\tr \D^{1/2} \D_{\piv}^{1/2} \ev)^2\right] \\&
= \max_{\x\tr\x = 1} \left[ (1{-}\mu)  \sumin w_i x_i^2  +  \mu \left(\sumin x_i \sqrt{\pi_i w_i} \right)^2 \right] \\&
\leq   (1{-}\mu)  \max_{\x\tr\x = 1} \sumin w_i x_i^2  +  \mu  \max_{\x\tr\x = 1}\left(\sumin x_i \sqrt {\pi_i w_i} \right)^2 \\&
= (1{-}\mu)  \max_{i}[w_i] +  \mu \sumin \pi_i w_i ,
}
the last term coming from Lemma \ref{Lemma:xmaxy}.  
\end{proof}


\section*{Counterexamples to Increased Mutational Robustness}

For neutral networks with non-symmetric mutation matrices, it is no longer true that population must concentrate on genotypes with above average mutational robustness.  Two counterexamples are provided.

\subsection*{Counterexample 1:  Cyclic Mutation}
Consider a space of three genotypes, two of which form a neutral network, with cyclic mutation
\ab{
\M = (1-\mu)\I + \mu \Bmatr{0 & 1 & 0\\0 & 0 & 1\\ 1&0&0}, \D=\Bmatr{1&0&0\\0&1&0\\0&0&0}.
}
The neutral network genotypes have average offspring fitness $\evt \D\M \D \ev / \evt \D \ev = 1- \mu/2 > r(\M \D) = 1-\mu.$

\subsection*{Counterexample 2:  Biased Reversible Mutation}
Consider a mutation matrix of three loci with two-alleles, under single-event mutation.  Suppose that mutation is biased towards a genotype with low mutational robustness.   Let genotype $1$ ($111$) be favored by mutation bias, where the transition matrix of the reversible Markov chain is
\ab{
\M=\left[
\begin{array}{cccccccc}
 0 & 3/4 & 3/4 & 0 & 3/4 & 0 & 0 & 0 \\
 1/3 & 0 & 0 & 1/3 & 0 & 1/3 & 0 & 0 \\
 1/3 & 0 & 0 & 1/3 & 0 & 0 & 1/3 & 0 \\
 0 &1/8 &1/8 & 0 & 0 & 0 & 0 & 1/3 \\
 1/3 & 0 & 0 & 0 & 0 & 1/3 & 1/3 & 0 \\
 0 &1/8 & 0 & 0 &1/8 & 0 & 0 & 1/3 \\
 0 & 0 &1/8 & 0 &1/8 & 0 & 0 & 1/3 \\
 0 & 0 & 0 & 1/3 & 0 & 1/3 & 1/3 & 0 \\
\end{array}
\right].
}
Genotypes $2,3,5$ ($101, 110, 011$) are its neighbors.  Genotypes $2,3,4$ ($101,110,100$) will be off the neutral network, so the fitnesses are
\ab{
\ev\tr \D = (1, 0, 0, 0, 1, 1, 1, 1).
}

Then the equilibrium population neutrality is less than the average neutrality of the neutral network:
\ab{
r(\M\D) = 0.6517 < 
\frac{\ev\tr \D \M \D \ev}{\ev\tr\D \ev} =  0.6667.
}

\pagebreak
\small 



\end{document}